%% file: ijcai26.tex
\documentclass{article}
\usepackage{ijcai26}

\usepackage{times}
\usepackage{soul}
\usepackage{url}
\usepackage[hidelinks]{hyperref}
\usepackage[utf8]{inputenc}
\usepackage[small]{caption}
\usepackage{graphicx}
\usepackage{amsmath}
\usepackage{amsthm}
\usepackage{booktabs}
\usepackage{algorithm}
\usepackage[noend]{algpseudocode} 
\usepackage[switch]{lineno}

\input{macros.tex}

\title{Ensuring Logic in the Fog: Sound POMDP Synthesis with LTL Objectives}

\author{
Can Zhou$^1$
\and
Yulong Gao$^1$\and
Pian Yu$^{2}$\\
\affiliations
$^1$Imperial College London, United Kingdom\\
$^2$University College London, United Kingdom\\
\emails
\{c.zhou24, yulong.gao\}@imperial.ac.uk,
pian.yu@ucl.ac.uk
}

\begin{document}

\maketitle

\begin{abstract}
Synthesising autonomous agents that can navigate uncertain environments while adhering to complex temporal constraints remains a fundamental challenge. 
While Linear Temporal Logic (LTL) provides a rigorous language for specifying such tasks, the inherent undecidability of qualitatively verifying LTL satisfaction in partially observable Markov decision processes renders quantitative synthesis difficult, especially when designing reliable reward signals for approximate solvers.
In this paper, we bridge this gap with a novel, sound reward-shaping mechanism that dynamically generates belief-dependent rewards grounded in certified LTL satisfaction. By integrating this mechanism into an enhanced Monte Carlo Planning framework, we empower agents to navigate the `fog' of partial observability with a search process focused on maximising verifiable success. Our experiments demonstrate that this approach not only thrives in scenarios where existing solvers fail but also maintains effectiveness and scalability across diverse benchmark domains.

\end{abstract}

\section{Introduction}
\label{sec:intro}
Partially Observable Markov Decision Processes (POMDPs) provide a principled framework for sequential decision-making amidst the fog of uncertainty, accounting for both stochastic dynamics and incomplete state information \cite{kaelbling1995partially,bongard2008probabilistic}. When coupled with $\omega$-regular properties like Linear Temporal Logic (LTL)~\cite{pnueli1977temporal}, they allow for the synthesis of policies that satisfy complex, temporally extended goals over infinite traces in uncertain, partially observable environments. Therefore, POMDP-LTL synthesis has attracted considerable interest across AI, robotics, and control~\cite{sharan2014finite,bouton2020point,kalagarla2024optimal}.

Despite its potential, quantitative POMDP-LTL synthesis, aiming to maximise satisfaction probability, remains notably challenging. Optimal policy synthesis for general planning tasks in infinite-horizon POMDPs is undecidable, even for simple reachability objectives, as it necessitates reasoning over an uncountable belief space and may require unbounded memory~\cite{chatterjee2016decidable}. Practical approaches therefore rely on tractable approximations that reformulate the problem as reward maximisation, with reward functions encoding task satisfaction, and apply approximate solvers such as Point-Based Value Iteration (PBVI)~\cite{pineau2003point} and Partially Observable Monte Carlo Planning (POMCP)~\cite{silver2010monte}.

However, constructing appropriate reward signals for full LTL is non-trivial, since its semantics concern non-terminating behaviour over infinite traces. Existing approaches either restrict attention to finite-trace LTL (LTL$_f$), which reduces satisfaction to a one-shot reachability problem~\cite{kalagarla2022optimal,liu2021leveraging}, or apply optimistic reward shaping based on state-space accepting maximal end components (AMECs)~\cite{bouton2020point}. Nevertheless, due to partial observability, POMDP policy synthesis must operate over the belief space, and many objectives such as recurrence and persistence require full LTL. Furthermore, an optimistic reward for POMDP-LTL may assign a positive satisfaction signal to a belief state even when no single belief-based policy (defined as a policy shared by all states within the belief’s support) can realise the implied satisfaction probability. This challenge, which we term the \emph{common policy issue}, is illustrated in the following example.


\begin{figure}
    \centering
    \includegraphics[width=0.5\linewidth]{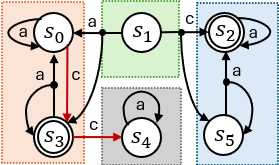}
    \caption{A POMDP where states sharing the same observation are indicated by a common background colour. This colour-coding convention is maintained throughout all subsequent figures. }
    \label{fig:motivation}
\end{figure}

\begin{example}
Consider the POMDP in Figure~\ref{fig:motivation} with action set $\{a, c\}$ and initial belief $b_0=\{s_1:1.0\}$. The LTL objective $\varphi=\ltlG\ltlF(s_2 \lor s_3)$, requiring the system to visit states $s_2$ or $s_3$ infinitely often. Under full observability, two AMECs exist: 1. $\{s_0, s_3\}$, which achieves almost-sure satisfaction (i.e. probability one) by alternating actions $a$ and $c$, and 2. $\{s_2, s_5\}$, which succeeds by repeatedly playing $a$. This 
suggests that all beliefs supported on AMECs yield satisfaction with probability one. However, with partial observability, no belief-based policy supported on $\{s_0, s_3\}$ can satisfy $\varphi$ almost surely: state $s_0$ requires action $c$, but taking $c$ in $s_3$ leads to the sink $s_4$. Conversely, beliefs supported on $\{s_2,s_5\}$ can still satisfy $\varphi$ almost surely by playing $a$.

Furthermore, assigning positive reward one to all states within the state-space AMECs (i.e., $\{s_0, s_2, s_3, s_5\}$) incorrectly suggests that actions $a$ and $c$ are equally optimal at $b_0$. In reality, $c$ is the uniquely optimal action.
\label{ex-motivation}
\end{example}

More fundamentally, encoding exact LTL satisfaction as a reward signal is  undecidable for POMDPs because of the probability leakage inherent in belief transitions, as shown in Example~\ref{ex-motivation} for beliefs supported on $\{s_0,s_3\}$. For such beliefs, deciding whether a policy with positive satisfaction probability exists is undecidable~\cite{baier2008decision}. These theoretical barriers make it difficult to construct sound reward signals. Without these, approximate solvers struggle to converge on reliable, belief-based policies for LTL objectives.

\textbf{Contributions.}
We bridge the theoretical and practical gaps in quantitative POMDP-LTL synthesis by introducing a novel mechanism for sound reward shaping. In contrast to fixed scalar reward structures, our approach defines reward assignment rules that integrate with an enhanced online POMCP planner. This allows the system to dynamically evaluate rewards and guide the search toward maximising certified satisfaction, thereby enabling reliable belief-based policy synthesis. To the best of our knowledge, this is the first sound and practical anytime synthesis algorithm for full LTL in POMDPs.
Our main contributions are summarised as follows: 1) \textbf{Sound Reward Shaping:} We introduce a sound reward shaping mechanism over the belief space by leveraging belief support structures, grounded in the insight that almost-sure satisfaction is the only certifiable property in POMDP-LTL. A pruning strategy streamlines reward construction, and the resulting rewards provide a certified lower bound on the true LTL satisfaction probability. 2) \textbf{Sound Anytime Synthesis:} We integrate our reward assignment rules into an enhanced POMCP algorithm specialised for anytime LTL synthesis. Once the search reaches a certified belief region, the agent may continue exploring for higher satisfaction or switch to a verified policy that guarantees satisfaction. This offers the first sound and practical method for real-time POMDP-LTL synthesis. 3) \textbf{Empirical Validation:} We show that our method synthesises reliable belief-based policies in scenarios where existing approaches fail. We further show its effectiveness and scalability on three standard POMDP benchmarks with diverse LTL tasks.

    
    

\textbf{Related Work.}
Owing to the undecidability of quantitative synthesis for POMDPs with LTL, prior work has turned to tractable approximations. One possible approach is to restrict policies to randomised finite-state controllers (FSCs)~\cite{chatterjee2015qualitative,sharan2014finite,ahmadi2020stochastic,carr2020verifiable,carr2021task}. However, such finite-memory restrictions impose key limitations: optimisation remains only locally optimal for any fixed memory size, and while performance is highly sensitive to that size, the minimal required memory is undecidable~\cite{madani1999undecidability}.  A complementary line of work~\cite{liu2021leveraging,kalagarla2022optimal,kalagarla2024optimal} considers belief-based policies for finite-trace LTL (LTL$_f$). However, many complex planning tasks require infinite-horizon recurrence and persistence, necessitating full LTL. \cite{bouton2020point} addresses this by applying PBVI with reward signals derived from state-based AMECs. As illustrated in Example~\ref{ex-motivation}, such state-based reward shaping can be overly optimistic under partial observability. A detailed discussion of related work is provided in Appendix~A.

\section{Preliminaries and Problem Statement}
This section first outlines the preliminaries of LTL and its equivalent Limit-Deterministic B\"{u}chi Automaton (LDBA) representation. We then provide the foundations of POMDPs and the POMCP algorithm for online planning. Finally, the problem under consideration is formulated. 

\subsection{Linear Temporal Logic}
LTL~\cite{pnueli1977temporal} extends propositional logic by introducing temporal operators that enable the specification of task requirements over time. The syntax of an LTL path formula over a finite set of atomic propositions ($\ap$) is defined inductively as follows:
\begin{equation*}\label{LTL}  
\varphi::= \ltltrue \mid  a \in \ap \mid \neg\varphi \mid \varphi \wedge\varphi \mid  \ltlX \varphi \mid \varphi \ltlU\varphi,
\end{equation*}
where $\ltlX$ (Next) and $\ltlU$ (Until) are temporal operators.
Additional derived temporal operators include: $\ltlF  \varphi \equiv \ltltrue \ltlU\varphi$ (Eventually) and $\ltlG \varphi \equiv \neg (\ltlF \neg\varphi)$ (Always). Standard Boolean operators, such as $\vee$ (or) and $\rightarrow$ (implies), are used conventionally. The detailed semantics with satisfaction relation of LTL can be found in~\cite{baier2008principles}.
For a formula $\varphi$ and an infinite word $w = w_0 w_1 w_2 \dots \in (2^{\ap})^\omega$, we write $\mathcal{L}(\varphi)$ to denote the language of $\varphi$, which is the set of infinite words over $(2^{\ap})^\omega$ that satisfy $\varphi$.

Every LTL formula admits an equivalent representation as a Nondeterministic B\"{u}chi Automaton (NBA)~\cite{vardi1994reasoning} that accepts the language~$\mathcal{L}(\varphi)$. More recently, a subclass of NBAs, known as limit-deterministic Büchi automata, has shown particularly suitable for the quantitative analysis of probabilistic systems~\cite{sickert2016limit}. In this paper, we also make use of LDBAs to support our approach.

\begin{definition}[LDBA~\cite{sickert2016limit}]
An LDBA $\mathcal{A}$ is defined as a tuple 
$\mathcal{A} = (\Sigma, Q, q_0, \Delta, \acc)$ where $\Sigma$ is an alphabet, $Q = Q_i \uplus Q_{acc}$ is the set of states partitioned into two disjoint sets $Q_i$ and $Q_{acc}$. $q_0 \in Q_i$ is the initial state. $\Delta = \Delta_i \uplus \Delta_{acc}$ where: $\Delta_i : Q_i \times (\Sigma \cup \{\epsilon\})  \to 2^{Q}$, $\Delta_{acc} : Q_{acc} \times \Sigma \to 2^{Q_{acc}}$ is singleton-valued, and $\acc \subseteq Q_{acc}$ is the accepting set.
\end{definition}

An infinite word $w \in \Sigma^\omega$ is recognised by an LDBA $\mathcal{A}$ if there exists an infinite run $r = r_0 r_1 r_2 \dots \in Q^\omega$, beginning from $q_0$ and satisfying $r_{i+1} \in \Delta(r_i, \omega_i)$ for all $i \geq 0$ as well as $\infset(r) \cap \acc \neq \emptyset$, where $\infset(r)$ denotes transitions in the accepting set appearing infinitely often in $r$.
The set of infinite words accepted by the LDBA $\mathcal{A}$ is denoted by $\mathcal{L}(\mathcal{A})$. An LTL formula $\varphi$ over $\ap$ can be efficiently translated into an equivalent LDBA $\mathcal{A}$ with $\Sigma = 2^{\ap}$, such that $\mathcal{L}(\mathcal{A}) = \mathcal{L}(\varphi)$, using state-of-the-art tools such as Owl~\cite{sickert2016limit} and Rabinizer4~\cite{kvretinsky2018rabinizer}.

\subsection{Partially Observable MDPs}
In this paper, we consider the following mathematical model:
\begin{definition}[Labelled POMDPs]
A \emph{labelled POMDP} is defined as a tuple $\M = (S, S_0, A, O, T, Z, L)$, where $S$, $S_0$, $A$, and $O$ are finite sets of states, initial states, actions, and observations, respectively. The state transition function $T : S \times A \to \dist(S)$, where $T(s, a)(s')$ denotes the probability of transitioning to state $s'$ from state $s$ by taking action $a$. The observation function $Z : S \times A \to \dist(O)$, where $Z(s', a)(o)$ denotes the probability of observing $o$ after taking action $a$ and reaching state $s'$. The labelling function $L : S \to 2^{\ap}$ maps each state $s \in S$ to the set of atomic propositions in $\ap$ that hold true at $s$.
\end{definition}


A belief state $b$ of $\M$ is a probability distribution over its state space $S$. 
Let $b_0$ denote the
initial belief state and let~$B$ be the set of all possible belief states for a given POMDP~$\M$. The \emph{belief support} of a belief state~$b \in B$ is the set of states with positive belief, defined as $\supp(b) := \{\, s \in S \mid b(s) > 0 \}.$  
The set of all possible belief supports in~$\mathcal{M}$ is then given by $B_{\supp} := \big\{\, \Theta \subseteq S \;\big|\; \exists\, o \in O,\; \forall\, s \in \Theta,\; o \in \obs(s) \big\},$ where $\obs : S \to 2^O$. 
We write $b(\Theta) := \sum_{s \in \Theta} b(s)$ for the probability mass assigned by $b \in B$ to $\Theta \subseteq S$, called its \emph{belief mass}.
The enabled action set at state $s$ is defined as $A_{\mathrm{E}}(s) := \{ a \in A \mid \post_s(a) \not = \emptyset \}$ where $\post_s(a) = \supp(T(s,a))$. We then define $A_{\mathrm{E}}(\Theta) := \bigcap_{s \in \Theta} A_{\mathrm{E}}(s)$.
Without loss of generality, we assume that states with the same observation admit the same set of enabled actions; that is, for any $s, s'\in \Theta$, $A_{\mathrm{E}}(s) = A_{\mathrm{E}}(s') = A_{\mathrm{E}}(\Theta)$.

\textbf{Paths.} A \emph{history} at time $t$ is a finite sequence $h_t = (a_0, o_1, \dots, a_{t-1}, o_t) \in (A \times O)^t$ capturing the agent's observable interaction with the environment. A \emph{belief path} is an infinite sequence $\rho = (b_0, a_0, o_1, b_1, a_1, o_2, \dots)$ in which each belief $b_{t+1} = \tau(b_t, a_t, o_{t+1})$ is obtained via Bayesian update. The set of such paths starting from $b_0$ is denoted by $\bpath(b_0)$. For any $\rho$, the set of \emph{compatible state paths} $\spath(\rho) \subseteq \mathcal{S}^\omega$ contains all sequences $(s_0, s_1, \dots)$ such that $s_0 \sim b_0$, $s_{t+1} \sim T(s_t, a_t)$, and $o_{t+1} \sim Z(s_{t+1}, a_t)$, consistent with the induced belief evolution.

\textbf{Policy and LTL Satisfaction.} We consider a deterministic belief-based policy $\pi \in \Pi: B \to A$ that maps a belief state to an action. Denote by $\Pi$ the set of these policies. Under a given policy $\pi$, the probability that an LTL specification $\varphi$ holds at belief state $b$ is denoted by 
$$\prob^\pi_b(\varphi) = \prob\{\rho_b \in \bpath(b), \rho_s \in \spath(\rho_b) \mid \rho_s \models \varphi\},$$
which measures all belief executions induced by $\pi$ whose compatible underlying state paths satisfy $\varphi$. Accordingly, the satisfaction probability for a POMDP $\M$ (under the policy $\pi$) with initial belief $b_0$ is defined as $\prob^\pi_\M(\varphi) := \prob^\pi_{b_0}(\varphi)$.

\textbf{Undecidability.} Optimal policy synthesis for POMDPs is generally undecidable under B\"uchi objectives induced by LTL. \cite{baier2008decision} show that, for B\"uchi objectives, both qualitative verification of positive (non-zero) satisfaction probability and quantitative verification are undecidable. The only property that can be certified is almost-sure satisfaction.

\textbf{Product POMDP.} Analogous to the product construction for MDPs~\cite{baier2008principles}, the product POMDP $\Mt$ synchronises the evolution of the POMDP $\M$ with the LDBA $\A_\varphi$ generated from the LTL $\varphi$. This allows satisfaction of $\varphi$ to be checked by tracking accepting runs of the LDBA along paths of $\M$.
\begin{definition}[Product POMDP]\label{productautomaton2}
Given a labelled POMDP $\M= (S, S_0, A, O, T, Z, L)$ and an LDBA $\A_\varphi= (\Sigma, Q, q_{0}, \trans, \acc)$ from the given LTL formula $\varphi$,
the product POMDP is a tuple $\Mt=(S^\times, S_0^\times, A^{\times}, O, T^\times, Z^\times, L^{\times}, \acc^\times)$ with state-space $S^\times=S \times Q$. The set of accepting states is given by $\acc^\times=\{(s, q)\in S^\times  \mid q\in \acc\}$. 
\end{definition}

Given the initial belief state $b_0$ of $\M$, we have $b^\times_0((s, \Delta(q_0, \mathcal{L}(s)))) = b_0(s)$, $\forall s \in S_0$. The set of all possible belief states of ~$\mathcal{M}^\times$ is denoted by $B^\times$. The complete definition of the product POMDP is provided in Appendix~B.

\subsection{Partially Observable Monte-Carlo Planning}
POMCP~\cite{silver2010monte} is a widely adopted online planning algorithm for POMDPs, as it effectively mitigates the \emph{curse of dimensionality} through state sampling and the \emph{curse of history} via history sampling with a black-box simulator. At each step $t$, it uses Monte Carlo Tree Search~\cite{coulom2006efficient} to build and explore a search tree rooted at $\T(h_t) = \langle \N(h_t), \V(h_t), \P(h_t) \rangle$, where $\N(h_t)$ tracks the number of visits to the history $h_t$, $\V(h_t)$ estimates the expected return, and $\P(h_t)$ holds particles that approximate the belief state $b_t$. 
Due to space limitations, details on how POMCP works are provided in Appendix~C.

\subsection{Problem Statement}
We study the following policy synthesis problem:
\begin{problem}
\label{p1}
Given a POMDP $\M$ and an LTL specification $\varphi$, synthesise a deterministic belief-based policy $\pi$ that maximises the probability of satisfying $\varphi$, i.e. $\max_{\pi\in \Pi}\  \Pr_{\M}^\pi(\varphi).$
\end{problem}

Following the standard product construction for LTL and probabilistic systems~\cite{baier2008principles}, we form the product POMDP
$\Mt=(S^\times, S_0^\times, A^\times, O, T^\times, Z^\times, L^\times, \acc^\times)$
of the POMDP $\M$ with the LDBA $\A_\varphi$ derived from $\varphi$ (Definition~\ref{productautomaton2}). This synchronises the executions of $\M$ and $\A_\varphi$ and reduces LTL satisfaction in $\M$ to the B\"uchi objective $\ltlG\ltlF\,\acc^\times$ in $\Mt$:
\[
\max_{\pi\in \Pi}\ \mathrm{Pr}_{\M}^\pi(\varphi)
=
\max_{\pi^\times\in \Pi^\times}\ \mathrm{Pr}_{\Mt}^{\pi^\times}(\ltlG\ltlF\,\acc^\times).
\]
Accordingly, we further reformulate Problem~\ref{p1} as an equivalent reward-based policy synthesis problem over $\Mt$. The goal is to design a belief-based reward function that induces an optimal policy whose expected reward is equal to the maximum satisfaction probability of B\"uchi objective.

\begin{problem}[Reward-Based Policy Synthesis]
\label{p2}
Given the product POMDP $\Mt$ induced by a POMDP $\M$ and an LTL objective $\varphi$, design a belief-based reward function $R: B^\times \to [0, 1]$ and synthesis a policy $\pi^\times: B^\times \to A^\times$ such that
\[
\max_{\pi^\times\in \Pi^\times}\ \mathbb{E}_{\Mt}^{\pi^\times}\!\left[\sum_{t=0}^{\infty} R(b_t^\times)\right]
=
\max_{\pi^\times\in \Pi^\times}\ \mathrm{Pr}_{\Mt}^{\pi^\times}(\ltlG\ltlF\,\acc^\times).
\]
\end{problem}

\section{Solution Technique}
Solving Problem~\ref{p2} presents significant theoretical hurdles. Encoding LTL satisfaction via a fixed scalar reward is infeasible since assigning an exact reward is equivalent to computing exact satisfaction probabilities, which is known to be undecidable (Section~\ref{sec:intro}). Moreover, even if such a reward were available, maximising the expected reward within a POMDP remains undecidable~\cite{madani1999undecidability}. Practical methods must therefore rely on approximations in both reward shaping and policy synthesis.

We address this gap by proposing a sound approximation framework designed for practical POMDP-LTL synthesis. Exploiting the fact that almost-sure satisfaction is the only certifiable property, we first show that the exact reward signal in Problem~\ref{p2} can be soundly approximated by the belief mass over certified satisfaction regions derived from the belief-support structure (Section~\ref{sec:reward-ap}). We then develop a belief-dependent reward-shaping mechanism that uses a set of assignment rules to generate sound rewards dynamically, together with a pruning strategy that streamlines reward construction (Section~\ref{sec:reward}). Finally, we integrate this mechanism into an enhanced POMCP algorithm that steers the search toward maximising certified satisfaction and synthesising reliable belief-based policies (Section~\ref{sec:pomcp}).

\subsection{Sound Approximation of Rewards Signal}
\label{sec:reward-ap}
For POMDPs with B\"uchi objectives, while determining exact or zero satisfaction is undecidable, almost-sure satisfaction fortunately admits a decidable property. Specifically, it depends only on the belief support rather than the precise probability distribution~\cite{chatterjee2007algorithms}: if a belief $b$ is almost-sure winning, then any belief $b'$ with $\supp(b')=\supp(b)$ is also almost-sure winning. 
This allows the continuous belief space to be discretised into a finite belief-support MDP (BSMDP), enabling structural evaluation of certified satisfaction over belief support.
In this abstraction, each state corresponds to a belief support $\Theta \subseteq 2^{S^\times}$, and transitions mimic the belief executions in the product POMDP.

\begin{definition}[BSMDP~\cite{junges2021enforcing}]
\label{def-bsmdp}
Given a product POMDP $\Mt=(S^\times, S_0^\times, A^{\times}, O, T^\times, Z^\times, L^{\times}, \acc^\times)$, let $B^\times_{\supp}$ be belief supports induced by $\Mt$. The corresponding BSMDP is the tuple $\Mt_{B_{\supp}} = (B^\times_{\supp}, \Theta_0, T^\times_{\supp}, A^\times)$, where $\Theta_0$ is the initial belief support, and the transition function $T^\times_{\supp}$ is defined as $\Theta' \in T^\times_{\supp}(\Theta, a \in A^\times_{\mathrm{E}}(\Theta))$ if
\begin{equation}
    \Theta'\in \left\{\cup_{s\in \Theta} \post_{s}(a)\cap \{s \mid o \in \obs(s)\} \mid o \in O^\times\right\}.
    \label{eq:beleif-update}
\end{equation}
\end{definition}

The BSMDP serves as a formal bridge between qualitative almost-sure satisfaction and general quantitative analysis.
By reasoning over the BSMDP, we can identify specific winning supports: subsets of state from which a belief-support based policy exists to satisfy the B\"uchi objective with probability 1. 
These winning supports enable the construction of a sound reward signal that, for any belief, lower-bounds the true satisfaction probability by the portion of belief mass guaranteed to satisfy the objective almost surely.
We now formalise the notion of winning supports and classify beliefs based on their underlying support structure.

\begin{definition}[Winning Support and Winning Region]
A subset $\Theta \subseteq S^\times$ is a winning support if there exists a belief-support policy $\hat{\pi}^\times: B^\times_{\supp} \to A^\times$ beginning with $\Theta$ such that every state $s^\times \in \Theta$ satisfies the LTL objective with probability one under $\hat{\pi}^\times$. We denote the union of all winning supports as a winning region $\Wt := \bigcup \Theta$.
\end{definition}

\begin{definition}[Belief Categorisation]
A belief $b^\times \in B^\times$ can be categorised by the relationship between its support $\supp(b^\times)$ and $\Wt$ as follows:
\begin{enumerate}
    \item \textbf{Almost-sure Winning:} if $\supp(b^\times) \in \Wt$;
    \item \textbf{Partially Winning:} if $\exists \Theta \in \Wt$ s.t. $\Theta \subset \supp(b^\times)$;
    \item \textbf{Non-Winning:} if $\forall \Theta \in \Wt$ and $\Theta \not \subseteq \supp(b^\times)$.
\end{enumerate}
\end{definition}

Note that even if the initial belief is non-winning, the agent may reach a (partially) winning belief during belief evolution.

\begin{figure}
    \centering
    \includegraphics[width=\linewidth]{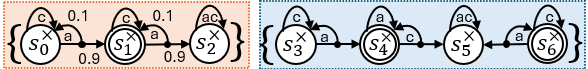}
    \caption{Partially winning belief (supports).}
    \label{fig:ex-max}
\end{figure}

\begin{example}
In the left of Figure~\ref{fig:ex-max}, beliefs supported on $\{s^\times_0, s^\times_1, s^\times_2\}$ are not almost-sure winning, but they contains a winning support $\{s^\times_1\}$ under action $c$. The satisfaction probability for any belief supported by $\{s^\times_0, s^\times_1, s^\times_2\}$ is lower-bounded by the belief mass assigned to $\{s^\times_1\}$ by repeating $c$.
\end{example}

Notably, winning supports are defined with respect to belief-support policies. Since a partially winning belief may contain multiple winning supports (each associated with a distinct policy), obtaining the tightest sound lower bound requires selecting the specific winning support that maximises the belief mass. Simply taking the union of all winning states would result in the \emph{common policy issue} (as discussed in Example \ref{ex-motivation}), where no single policy can satisfy the entire set. This principle is illustrated in the following example.

\begin{example}\label{exp3}
Consider beliefs supported by $\{s_3^\times, s_4^\times, s_5^\times, s_6^\times\}$ in the right of Figure~\ref{fig:ex-max}. The winning supports and their admissible actions are: $\{s^\times_3\}$ with $\{a\}$, $\{s^\times_4\}$ with $\{a\}$, $\{s^\times_3,s^\times_4\}$ with $\{a\}$, and $\{s^\times_6\}$ with $\{c\}$. These cannot be merged into $\{s^\times_3,s^\times_4,s^\times_6\}$ when computing a sound lower bound, as no single policy is winning for all three states. The tightest lower bound is therefore determined by whichever has the larger belief mass: $\{s^\times_3,s^\times_4\}$ or $\{s^\times_6\}$.
\end{example}

We now define $K_{b^\times}$ as the collection of all winning supports contained within the support of a belief state $b^\times$ in the product POMDP $\Mt$. For instance, in Example \ref{exp3}, a belief $b^\times$ with support $\{s_3^\times, s_4^\times, s_5^\times, s_6^\times\}$ would have the set of winning supports $K_{b^\times}=\{\{s^\times_3\}, \{s^\times_4\}, \{s^\times_3, s^\times_4\}, \{s^\times_6\}\}$. Then a sound belief-based reward signal can be defined as the maximum belief mass assigned to any winning support, reducing Problem~\ref{p2} to the following sound approximation:


\begin{theorem}\label{thm1}
    Given a product POMDP $\Mt$, the maximum expected reward in Problem~\ref{p2} is lower-bounded by a belief-based policy $\tilde{\pi}^\times: B^\times \rightarrow A^\times$ that maximises the probability of reaching a (partially) winning belief $b^\times$ with maximal belief mass over its winning supports $K_{b^\times}$:
    \[
    \max_{\pi^\times\in \Pi^\times}\ \mathbb{E}_{\Mt}^{\pi^\times}\!\left[\sum_{t=0}^{\infty} R(b_t^\times)\right]
    \geq 
    \max_{\tilde{\pi}^\times \in \Pi^\times} \mathbb{E}_{\Mt}^{\tilde{\pi}^\times} \left[ \max_{\Theta \in K_{b^\times}} b^\times(\Theta) \right],
    \]
    \label{thm:lb}
    where $b^\times(\Theta)$ is the belief mass assignment to the set of states contained in $\Theta$.
\end{theorem}


\textbf{Tightness.} The lower bound in Theorem~\ref{thm:lb} attains the maximum certified LTL satisfaction probability. The remaining gap to the true value reflects satisfaction probability that cannot be certified by any belief-support policy, arising from the undecidable leakage illustrated in Example~\ref{ex-motivation}.

\subsection{Sound Reward Shaping}
Theorem~\ref{thm:lb} allows sound reward signals to be constructed by identifying winning supports in the BSMDP.
However, the state space of BSMDP grows exponentially with the state space of the product POMDP, rendering all winning supports and their associated winning policies computationally expensive. To maintain tractability, we extract a sufficient subset of winning supports, together with their almost-sure or partially winning beliefs, by pruning the search space according to two principles: (i) restricting the state space of the BSMDP to obtain a sub-BSMDP; and (ii) focusing on structurally relevant regions within this sub-BSMDP. A sound reward-shaping mechanism is then derived from this refined subset.

\textbf{Pruned Search Regions.}  Recall that in fully observable product MDPs with a B\"uchi objective, any satisfying run eventually enters and remains within a state-based AMEC. We hence restrict the BSMDP construction to belief supports that intersect the set of underlying state AMECs $\E^\times_S$ in $\Mt$:

\begin{equation}
    \hat{B}_\supp^{\times}:= \{\Theta \in B^{\times}_{\supp} \mid \Theta \cap \E^\times_S \not= \emptyset\}.
    \label{eq:sub}
\end{equation}

This restriction gives an over-approximation of the belief-support states relevant for satisfaction under belief-based policies, while substantially reducing the BSMDP state space. Moreover, since every winning support ensures almost-sure satisfaction, any belief-support path must almost surely enter a trapping region of the BSMDP in which each support state recurs infinitely often with satisfaction probability one. In the BSMDP, such closed recurrent accepting regions correspond to belief-support accepting MECs (BS-AMECs), defined formally below.

\begin{definition}[BS-(A)MECs]
A belief-support end component (BS-EC) of an BSMDP $\Mt_{B_\supp}$ is a sub-MDP $(E, A_E)$ such that the digraph induced by $(E, A_E)$ is strongly connected. A belief-support MEC (BS-MEC)  is an BS-EC $(E, A_E)$ that is not strictly contained in any other BS-EC $(E', A'_E)$.
We call a BS-MEC $(E, A_E)$ accepting (BS-AMEC) if all belief supports $\Theta \in E$ are winning. 
\end{definition}

\begin{figure}
    \centering
    \includegraphics[width=\linewidth]{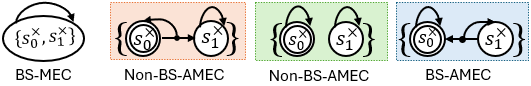}
    \caption{BS-MEC structure under differing acceptance conditions.}
    \label{fig:bsmecs}
\end{figure}
\label{sec:reward}

\textbf{Sufficient Winning Region.}
We therefore restrict the search for the set of winning-supports $\hat{\W}^\times$ ($\hat{\W}^\times\subseteq \W^\times$) to those contained in the BS-AMECs of the sub-BSMDP $\hat{\M}^\times_{B_{\supp}}$, constructed over belief supports that intersect the state-space AMECs defined in Eq.~(\ref{eq:sub}). Any other winning support will almost surely reach $\hat{\W}^\times$ (thus $\hat{\W}^\times$ is sufficient).

Determining, however, whether a BS-MEC is accepting is non-trivial, since the presence of $\acc^{\times}$ alone is insufficient. As illustrated in Figure~\ref{fig:bsmecs}, a BS-MEC may contain no winning support, may be only partially accepting, or is a BS-AMEC (from left to right in Figure~\ref{fig:bsmecs}).
Hence, belief-support information alone is inadequate: one must verify the existence of a belief-support-based policy under which all underlying states are winning. 
To this end, Algorithm~\ref{alg:amec} identifies BS-AMECs by constructing an auxiliary product MDP $\M^{\times\times}=(S^{\times\times}, A^{\times\times}, T^{\times\times}, \acc^{\times\times})$, where $S^{\times\times}= \{(s^\times,\Theta) \mid \ s^\times \in \Theta,\Theta \in E\}$, $A^{\times\times}(s^\times,\Theta)=A_{\E^\times}(\Theta)$, $T^{\times\times}((s^\times,\Theta),a)((s'^\times,\Theta'))=T^\times(s^\times,a)(s'^\times)$, and $\acc^{\times\times}= \{(s^\times,\Theta)\in S^{\times\times} \mid s^\times \in \acc^\times\}$. One can see that $\M^{\times\times}$ couples each BS-MEC with the underlying state-space dynamics, permitting joint reasoning over belief-support policies and state transitions. A BS-MEC is accepting iff all states in this product MDP are winning.

\begin{theorem}\label{thm2}
Given the sub-BSMDP $\hat{\M}^{\times}_{B_{\supp}}$ and its BS-MECs $(\E^{\times}_{B_{\supp}}, A_{\E^{\times}})$, Algorithm~\ref{alg:amec} returns a sound and complete union set $\hat{\W}^\times$ over all BS-AMECs.  
\end{theorem}

\begin{proofscratch}
\emph{Soundness.}
If there exists $\Theta \in E$ with $\Theta \subseteq \acc^\times$, then $E$ is a winning support and $(E,A_E)$ is a BS-AMEC by the MEC property.
Otherwise, consider the randomised belief-support policy $\bar{\pi}^{\times}$ that selects uniformly from $A_E$.
If $S^{\times\times}_w = S^{\times\times}$, every BSCC induced by $\bar{\pi}^{\times}$ is accepting.
By Lemma~1 of~\cite{chatterjee2010randomness}, a corresponding deterministic policy $\tilde{\pi}^{\times}$ exists; hence $(E,A_E)$ is a BS-AMEC.
\emph{Completeness.}
Immediate from Definitions~5 and~7.
A full proof is provided in Appendix~D.
\end{proofscratch}





\begin{algorithm}
\caption{Compute BS-AMECs}
\label{alg:amec}
\begin{algorithmic}[1]
\algrenewcommand\algorithmicrequire{\textbf{ Input:}}
\algrenewcommand\algorithmicensure{\textbf{ Output:}}

\Require BS-MECs $(\E^{\times}_{B_{\supp}}, A_{\E^{\times}})$ of sub-BSMDP $\hat{\M}^{\times}_{B_{\supp}}$
\Ensure Sufficient winning region $\hat{\W}^\times$

\State $\hat{\W}^\times \gets \emptyset$

\ForAll{$(E, A_E) \in (\E^{\times}_{B_{\supp}}, A_{\E^{\times}})$}
    \If{$\exists\, \Theta \in E \text{ s.t. } \Theta \subseteq \acc^\times$}
        \State $\hat{\W}^\times \gets \hat{\W}^\times \cup \{\Theta \in E\}$
    \Else
        \State Construct the product MDP $\M^{\times\times}$
        \State Compute the winning state set $S^{\times\times}_w$ of $\M^{\times\times}$
        \If{$S^{\times\times}_w = S^{\times\times}$}
            \State $\hat{\W}^\times \gets \hat{\W}^\times \cup \{\Theta \in E\}$
        \EndIf
    \EndIf
\EndFor

\end{algorithmic}
\end{algorithm}

\textbf{Sound Reward Assignment Rules.} 
We then define the sound reward structure
based on the sufficient winning region $\hat{\W}^\times$ computed by Algorithm~\ref{alg:amec}. We call a belief support $\Theta\in \hat{\W}^\times$ a \emph{certified winning support}.
Furthermore, $\Theta$ is called a \emph{certified partially winning support} if there exists a certified winning support $\Theta' \in \hat{\W}^\times$ such that $\Theta' \subset \Theta$. The set of all certified partially winning support is denoted by $\hat{\W}^\times_{\text{PW}}$.
Since a belief $b^\times$ with certified partially winning support $\supp(b^\times)\in \hat{\W}^\times_{\text{PW}}$ may contain multiple certified winning supports (see Example \ref{exp3}), we let $\hat{K}_{b^\times}$ be the collection of all such supports contained within $\supp(b^\times)$. 


Following the lower bound established in Theorem \ref{thm1}, the sound reward function $\hat{R}(\cdot)$ is defined as:
\begin{equation}\label{reward}
    \hat{R}(b^\times) =
    \begin{cases}
        1, & \text{if } \supp(b^\times)\in \hat{\W}^\times, \\[2mm]
        \displaystyle \max_{\Theta \in \hat{K}_{b^\times}} b^\times(\Theta), 
        & \text{if } \supp(b^\times)\in \hat{\W}^\times_{\text{PW}}, \\[3mm]
        0, & \text{otherwise},
    \end{cases}
\end{equation}
 where $b^\times(\Theta)$ is the belief mass assignment to the set $\Theta$.

\subsection{Online POMDP Planning}
\label{sec:pomcp}
To synthesise reliable belief-based policies for POMDP-LTL using the proposed sound reward function $\hat{R}(\cdot)$ in Eq. (\ref{reward}), we develop an enhanced POMCP algorithm as an anytime approximation method. Standard POMCP, however, is not directly applicable, as it relies on Monte Carlo simulations that assign rewards to individual state particles, which is incompatible with our belief-based reward signal. The termination conditions for the Monte Carlo simulation also require rigorous treatment. While entering a belief with a certified winning support is terminal as it already yields the maximum reward, entering a belief with a certified partially winning support is not necessarily so. In the latter case, the agent may continue to explore beliefs that offer higher satisfaction probabilities, a distinction further illustrated in Example~\ref{ex-term}.


\begin{example}
\label{ex-term}
Consider a belief node $b^\times$ in the Monte Carlo search tree with ideal distribution $\{s^\times_0:0.8, s^\times_1:0.1, s^\times_2:0.1\}$ in the left of Figure~\ref{fig:ex-max}. It is clear that $b^\times$ is partially winning since its support contains a winning support $\{s_1^\times\}$. If simulations terminate whenever a particle reaches any (certified) winning support, the resulting satisfaction probability of $b^\times$ is $0.1$ by the belief mass on $\{s^\times_1\}$. 
However, if simulations do not terminate early and all particles continue exploring, taking action $a$ leads to the belief $\{s^\times_0:0.08, s^\times_1:0.73, s^\times_2:0.19\}$, which yields a higher lower bound as $0.73$.
\end{example}

We address reward assignment for particles by augmenting each POMCP tree node $\mathcal{T}(h_{t+k})$, which represents the belief $b_{t+k}^\times$ at time $t$ and simulation depth~$k$, with its ground-truth belief support $\Theta(h_{t+k}) = \supp(b_{t+k}^\times)$. This support is directly recoverable from the history.
Intuitively, the search tree simulates BSMDP: each particle path induces a belief-support trace determined only by the history, allowing belief supports to be recovered without explicitly tracking belief distributions.
Hence, for each particle $s^\times_{t+k}$, we assign reward 1 if the associated belief support $\Theta(h_{t+k})$ is in $\hat{\W}^\times$, or if it is in $\hat{\W}^\times_{\text{PW}}$ and the particle lies in a certified winning support $\Theta \in \hat{K}_{b^\times_{t+k}}$. All remaining particles receive reward 0.



To handle termination for a POMCP tree node with a certified partially winning support, the Upper Confidence Bound (UCB) rule evaluates both terminal and exploratory options. Consider, for instance, a belief $b^\times$ with support $\{s^\times_3, s^\times_4, s^\times_5, s^\times_6\}$, as shown on the right of Figure~\ref{fig:ex-max}. As discussed in Example \ref{exp3}, the tightest lower bound for satisfaction is determined by the maximum belief mass assigned to either $\{s^\times_3, s^\times_4\}$ or $\{s^\times_6\}$. Consequently, the UCB selection mechanism identifies the optimal path by switching between: 1) the terminal policy for $\{s^\times_3, s^\times_4\}$ if it carries the dominant mass; 2) the terminal policy for $\{s^\times_6\}$ if it is more probable; or 3) any available exploration actions if they promise higher satisfaction through further search.
Exploration proceeds until a certified winning support is reached or the maximum depth is exceeded. This mechanism allows the planner to terminate at any time during execution by selecting an optimal belief-support policy; the resulting satisfaction probability is then guaranteed by the lower bound $\max_{\Theta \in \hat{K}_{b^\times}} b^\times(\Theta)$. The full planning algorithm are given in Appendix~E.2.

\section{Experiments}
We implemented our framework in Python\footnote{The source code is provided in https://github.com/Safe-Autonomy-and-Intelligence-Lab/POMDP-LTL}. In the following, We first evaluate a motivating toy example demonstrating that existing methods fail even in simple cases, while our approach succeeds in synthesising reliable belief-based policies. We further assess the effectiveness of our method on three standard POMDP benchmarks under various LTL tasks. The implementation details, experimental setup and additional results are provided in Appendix~E and F.

\subsection{Motivating Toy Example}
\begin{figure}[t]
    \centering
    \includegraphics[width=0.48\linewidth]{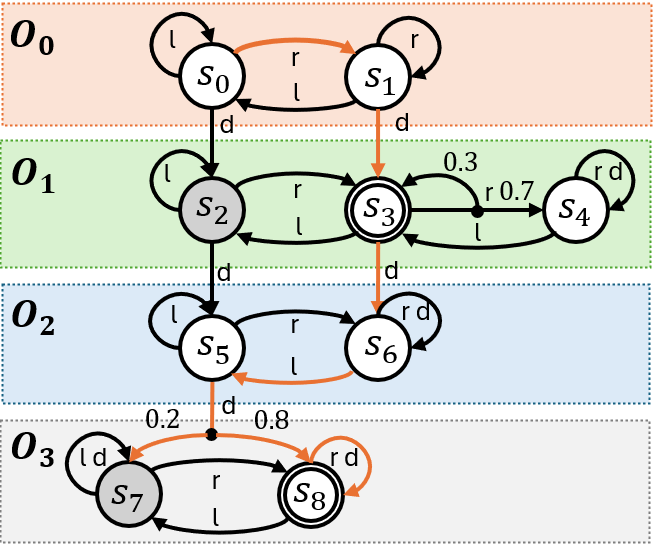}
    \caption{Motivating Toy Example.}
    \label{fig:toy}
\end{figure}

Consider the POMDP in Figure~\ref{fig:toy} with initial belief $b_0=\{s_0:1\}$ and actions left ($l$), right ($r$), and down ($d$). The LTL specification is
$
\varphi = \ltlG \neg(s_2 \lor s_7)\ \land\ \ltlG \ltlF (s_3 \lor s_8),
$
which requires visiting $s_3$ or $s_8$ infinitely often while always avoiding $s_2$ and $s_7$.
We compare our approach against optimistic reward shaping~\cite{bouton2020point} and a deterministic FSC variant~\cite{sharan2014finite}.

\textbf{Our method.}
Our analysis identifies a single winning support, $\{s_8\}$. The sound reward shaping assigns to each belief a value equal to its belief mass on $\{s_8\}$. Running our solver obtains the following belief-based policy and history: $b_0:\{s_0:1\}$ $\xrightarrow{r}$ $b_1:\{s_1:1\}$ $\xrightarrow{d}$ $b_2:\{s_3:1\}$ $\xrightarrow{d}$ $b_3:\{s_6:1\}$  $\xrightarrow{l}$ $b_4:\{s_5:1\}$  $\xrightarrow{d}$ $b_5:\{s_7:0.2, s_8:0.8\}$ $\xrightarrow{\{s_8\}}$. At $b_5$, the policy switches to the belief-support policy for $\{s_8\}$ (e.g. alternating $r$ and $d$), ensuring $\prob^{\pi}_{b_0}(\varphi)\ge 0.8$. Experimental details and POMCP values are provided in Appendix~F.2.

\textbf{Existing approaches.}
Using the optimistic reward shaping in~\cite{bouton2020point} with their PBVI solver reports a satisfaction probability of $1$ and forces early termination at the belief $b=\{s_3:1\}$, preventing sufficient exploration of the reachable belief space. Hence, the solver fails to construct $\alpha$-vectors valid for successor beliefs. Since the only exit actions encoded in these $\alpha$-vectors are $r$ and $d$, with $l$ absent for any possible satisfaction, execution ultimately results in zero LTL satisfaction. 
We next consider the deterministic one-memory FSC variant in~\cite{sharan2014finite}, which collapses to a memoryless observation-based policy. However, no FSC provides positive satisfaction for $b_0$. For observation $o_0$, action $r$ leads to $\{s_1:1\}$, which self-loops under $r$; action $d$ violates $\ltlG\neg(s_2\lor s_7)$; and action $l$ keeps the system at $\{s_0:1\}$. Larger memory may help, but selecting a suitable memory size is undecidable, making FSC synthesis brittle in practice. Conversely, our deterministic belief-based policy remains executable and offers a meaningful certified lower bound.

\subsection{Benchmark Domains and LTL Objectives}

\textbf{Hallway.}
We evaluate two Hallway models, $HW_1$ and $HW_2$, from~\cite{littman1995learning}. The agent moves in an office environment with uncertain location and four orientations, using stochastic forward and left turn actions. A short range sensor reports adjacent walls with some error. Three locations are labelled $A$, $B$ and $C$. We consider two LTL objectives: $\varphi_1 = \ltlF A \land \ltlG \neg B$, requiring the agent to eventually reach $A$ while avoiding $B$ at all times, and $\varphi_2 = \ltlF \ltlG C$, requiring the agent to reach $C$ and remain there indefinitely.

\noindent\textbf{Rock Sample.}
We follow the rock sample model from~\cite{junges2021enforcing}, where a robot moves on an $n \times n$ grid with two rocks, each either good $G$ or bad $B$ with initially unknown quality. The robot can move, sense adjacent cells, and sample rocks. We consider two LTL tasks: $\varphi_3 = \ltlF(G \land \ltlF E) \land \ltlG \neg B$, requiring it to first sample a good rock, then reach in an exit state $E$ while never reaching or sampling a bad rock, and $\varphi_4 = \ltlF \ltlG G$, requiring the robot to collect a good rock and remain there indefinitely.

\noindent\textbf{Grid World.}
The $Grid(n, n)$ model from~\cite{littman1995learning} is an $n \times n$ grid where each cell is labelled as a trap ($T$), goal 1 ($G_1$), or goal 2 ($G_2$). The agent moves left, right, or down, with a 0.2 chance of movement failure, and receives observations of neighbouring cells. We consider three LTL tasks. The first is $\varphi_5 = \ltlF G_1 \land \ltlG \neg T$, requiring the agent to eventually reach $G_1$ while always avoiding traps. The second is $\varphi_6 = \ltlG (G_2 \rightarrow \ltlF G_1) \land \ltlG \neg T$, requiring it to always avoid traps and, whenever it reaches $G_2$, eventually reach $G_1$.
The third is $\varphi_7 = \ltlG \ltlF G_1 \land \ltlG \neg T$, requiring it to visit $G_1$ infinitely often while avoiding traps.

\begin{table}[ht]
\centering
\resizebox{\linewidth}{!}{
\begin{tabular}{rcccc}
\toprule
 Domain & $|S^\times| / |O^\times| / |T^\times|$ & SR (s) & POMCP (s) & $\Pr_{\M}^\pi(\varphi)$  \\

\midrule
\multirow{5}{*}{}
\textbf{HW} \\
HW1 $\varphi_1$ & 180 / 24 / 624 & 0.015 & 74.967 & 0.254 \\
$\varphi_2$ & 180 / 24 / 684 & 0.020 & 52.759 & 0.707 \\
HW2 $\varphi_1$ & 276 / 24 / 975  & 0.547 & 64.356 & 0.247  \\
 $\varphi_2$ & 276 / 24 / 1067 & 0.743 & 44.643 & 0.464 \\

\midrule
\multirow{7}{*}{}
\textbf{RS} \\
4 $\varphi_3$ & 1024 / 112 / 7936 & 42.231 & 79.102 & 0.453  \\
$\varphi_4$ & 768 / 112 / 6208  & 0.078 & 91.788 & 0.469  \\
 
5 $\varphi_3$ & 1600 / 132 / 12480 & 58.246 & 121.315 & 0.326 \\
 $\varphi_4$ & 1200 / 132 / 9760  & 0.114 & 120.307 & 0.455 \\
 
7 $\varphi_3$ & 3136/ 244 / 24640 & 105.475 & 129.259 & 0.180  \\
$\varphi_4$ & 2352 / 244 / 19264 & 0.210 & 160.291 & 0.003  \\

\midrule
\multirow{10}{*}{}
\textbf{Grid} \\
50 $\varphi_5$ & 7500 / 1250 / 30990 & 10.987 & 82.198 & 0.334 \\
 $\varphi_6$ & 7500 / 1250 / 30990 & 10.597 & 0.246 & 1.000  \\
$\varphi_7$ & 5000 / 1250 / 20660 & 1.843 & 67.959  & 0.530 \\

100 $\varphi_5$ & 30000 / 5000 / 120840 & 81.361 & 99.407 &0.272  \\
$\varphi_6$ & 30000 / 5000 / 120840 & 82.334 & 0.248 & 1.000  \\
$\varphi_7$ & 20000 / 5000 / 80560 & 30.603 & 107.366 & 0.308  \\

120 $\varphi_5$ & 43200 / 7200 / 173580 & 191.431 & 182.667 & 0.000 \\
$\varphi_6$ & 43200 / 7200 / 173580 & 188.264 & 0.267 & 1.000  \\
$\varphi_7$ & 28800 / 7200 / 115720  & 66.314 & 183.521 & 0.058  \\

\bottomrule
\end{tabular}}
\caption{Numerical Results for POMDP-LTL Policy Synthesis.}
\label{res-table}
\end{table}

\subsection{Results and Discussion}
Table~\ref{res-table} summarises the performance of our approach across all tasks. For each case, we report the product model details ($|S^\times| / |O^\times| / |T^\times|$), the LTL objective ($\varphi_i$) with its satisfaction probability ($\Pr_{\M}^\pi(\varphi)$), the time to construct the sound reward function (SR (s)), and the policy synthesis time using the enhanced POMCP solver (POMCP (s)). Further details on the original models and additional experimental results are provided in Appendix~F.4.

The results show that our method synthesises POMDP-LTL policies effectively and efficiently, scaling to product models with over forty thousand states, with both sound reward construction and planning completing in about 190 seconds on a single CPU thread. Our approach is particularly effective for LTL tasks such as the reach-persistence objective $\varphi_4$ and the recurrence objective $\varphi_7$. Although the underlying state-space AMECs grow proportionally with the model size, they remain confined to a small region of the overall state space, which enables efficient construction of sound reward signals. For example, in the Rock Sample domain, as the model scales, the construction time for the task $\varphi_4$ increases only from $0.078$ to $0.210$ seconds.

For the POMCP results, although theory ensures $\prob_{\M}(\varphi_5) \ge \prob_{\M}(\varphi_7)$ as every trace satisfying $\varphi_5$ also satisfies $\varphi_7$, we observed $\prob_{\M}(\varphi_5) < \prob_{\M}(\varphi_7)$ in practice. This mismatch arises from the differing sizes of their LDBAs: $|\A_{\varphi_5}| = 3$ versus $|\A_{\varphi_7}| = 2$. The larger automaton for $\varphi_5$ results a larger product POMDP, requiring more simulations and deeper search for similar approximation quality.

\section{Conclusion}
This paper presented the first sound and practical anytime belief-based policy synthesis algorithm for full LTL in POMDPs. We approximate undecidable quantitative POMDP-LTL synthesis by seeking policies that maximise verifiable success. Central to this approach is a sound belief-dependent reward-shaping mechanism, integrated into an enhanced POMCP algorithm for reliable policy synthesis. 
For future work, we plan to extend this sound reward mechanism to a model-free setting, enabling safe POMDP reinforcement learning under LTL objectives.

\bibliographystyle{named}
\bibliography{main}

\input{appendix}

\end{document}

%% file: macros.tex
\usepackage{amsmath, amsthm, amssymb}
\usepackage{nccmath}
\usepackage{xspace}
\usepackage[dvipsnames]{xcolor}
\usepackage{todonotes}
\usepackage{tikz}
\usepackage{booktabs}
\usepackage{multirow}
\usetikzlibrary{automata, positioning, arrows.meta, angles, quotes, calc}
\usepackage{graphicx}
\usepackage{subcaption}
\usepackage{listings}

\tikzset{every picture/.style={line width=0.5pt}} 

\newtheorem{definition}{Definition}

\newtheorem{problem}{Problem}

\newtheorem{example}{Example}
\newtheorem{theorem}{Theorem}
\newenvironment{proofscratch}
  {\begin{proof}[Proof scratch]}
  {\end{proof}}

\usepackage{acro}

\newcommand{\bpath}{\text{BPaths}}
\newcommand{\spath}{\text{SPaths}}

\newcommand{\A}{\mathcal{A}}

\newcommand{\E}{\mathcal{E}}

\newcommand{\M}{\mathcal{M}} 
\newcommand{\N}{\mathcal{N}}
 
\renewcommand{\P}{\mathcal{P}}

\newcommand{\T}{\mathcal{T}}
 
\newcommand{\V}{\mathcal{V}}
\newcommand{\W}{\mathcal{W}}

\newcommand{\Mt}{\mathcal{M}^\times} 
\newcommand{\Wt}{\mathcal{W}^\times}

\newcommand{\trans}{\Delta}
\newcommand{\acc}{\mathtt{Acc}}

\newcommand{\ap}{\mathtt{AP}}
\newcommand{\supp}{\mathtt{Supp}}
\newcommand{\post}{\mathtt{Post}}

\newcommand{\dist}{\mathtt{Dist}}
\newcommand{\obs}{\mathtt{obs}}
\newcommand{\infset}{\mathtt{Inf}}

\newcommand{\ltltrue}{\mathsf{true}}

\newcommand{\ltlU}{\mathsf{U}}
\newcommand{\ltlF}{\lozenge}
\newcommand{\ltlG}{\square}
\newcommand{\ltlX}{\bigcirc}

\newcommand{\prob}{\rm{Pr}}

%% file: appendix.tex
\appendix

\section*{Appendix}
\section{Extended Related Work}
\label{ap-rw}
For fully observable MDPs, LTL satisfaction is identified structurally via Accepting Maximal End Components (AMECs) within the product system. Under B\"uchi acceptance conditions, a state-based policy (e.g., a Round-Robin policy) can ensure almost-sure satisfaction. Consequently, policy synthesis for MDPs under LTL is effectively reduced to a reachability problem toward AMECs, solvable via value iteration or linear programming~\cite{baier2008principles}.

In contrast, the general problem of finding a policy to satisfy an LTL formula in an infinite-horizon POMDP is undecidable. Prior work therefore has pursued various tractable approximations. A possible approach is to restrict the policy class to randomised finite-memory policies encoded as finite-State controllers (FSCs). While~\cite{chatterjee2015qualitative} provides heuristic procedures for qualitative (almost-sure) objectives, quantitative synthesis, seeking to maximise the satisfaction probability, typically reduces to a computationally intensive non-linear or relaxed optimisation problem~\cite{sharan2014finite,ahmadi2020stochastic}. Further model-free analyses of FSCs are given in~\cite{carr2020verifiable,carr2021task}.

However, restricting the policy class introduces fundamental limitations. Probability-leakage issue, such as those in Example~1, are inherently overlooked by finite memory, and optimisation over FSCs generally yields only locally optimal solutions for the chosen memory size~\cite{ahmadi2020stochastic}. Moreover, performance is highly sensitive to this size, but determining the minimal memory required for a given task is undecidable~\cite{madani1999undecidability}. Small FSCs often fail to retain essential history, resulting in zero satisfaction probability~\cite{sharan2014finite}. 
Additionally, existing approaches further rely on randomised policies to facilitate continuous optimisation, but in practice, particularly in safety-critical systems, such stochasticity undermines reproducibility. Therefore, there is a need for deterministic belief-based policies, mapping each belief state to an action, to ensure consistent and predictable execution.

A complementary line of work focuses on finite-trace LTL specifications (LTL$_f$) under belief-based policies, yielding a one-shot satisfaction formulation. Under LTL$_f$ semantics, planning reduces to reaching accepting states of a deterministic finite automaton in the product system, enabling the use of standard approximate POMDP solvers; e.g.,~\cite{kalagarla2022optimal,kalagarla2024optimal} employs PBVI and~\cite{liu2021leveraging} uses POMCP.  However, many planning tasks impose infinite-horizon constraints, such as safety, recurrence, and persistence, that require full LTL semantics. To address this,~\cite{bouton2020point} considers full LTL and applies PBVI to maximise satisfaction probability. A limitation is that, when state-based accepting MEC structure is used for reward shaping within an approximate solver, the induced value estimates can be optimistic under partial observability: the shaping signal can suggest a non-trivial satisfaction probability for a belief, without providing a corresponding certified belief-based policy that realises this probability by steering the system to and sustaining accepting behaviour (see Example 1). This motivates the development of sound reward signals that provide certified LTL satisfaction guarantees and guide approximate solvers towards synthesising sound deterministic belief-based policies for full LTL in POMDPs.

\section{Product POMDP}\label{product_pomdp}
The complete definition of the product POMDP is given below.

\begin{definition}[Product POMDP]
Given a labelled POMDP $\M= (S, S_0, A, O, T, Z, L)$ and an LDBA $\A_\varphi= (\Sigma, Q, q_{0}, \trans, \acc)$ from the given LTL formula $\varphi$,
the product POMDP is a tuple $\Mt=(S^\times, S_0^\times, A^{\times}, O, T^\times, Z^\times, L^{\times}, \acc^\times)$, where:
\begin{itemize}
    \item $S^\times=S \times Q$ is the State space with initial States $S_0^\times = \{(s, \Delta(q_0, L(s))) \in S^\times \mid s \in S_0 \}$, and $A^{\times}=A\cup A^\epsilon$ where $ A^\epsilon:=\{\epsilon_q\mid q\in Q\}$;
    \item $T^{\times}: S^\times \times A^{\times} \to \dist({S^\times})$ is the probabilistic transition function, where :
    \begin{itemize}
        \item for $a \in A$, $T^{\times}((s, q), a)((s', q')) = {T}(s, a)(s')$ if $a\in A_E(s)$ and $q'= \trans(q, {L(s')})$;
        \item for $ a = \epsilon_{q'}\in A^\epsilon$, $T^{\times}((s, q), a)((s, q')) = 1$ {if $q \xrightarrow{\epsilon} q'$};
        \item otherwise, $T^{\times}((s, q), a, (s', q')) = 0$;
    \end{itemize}
    \item $Z^\times: S^\times \times A^\times \to \dist(O)$ is the observation function, given by $Z^\times((s, q), a, o) = Z(s, a, o)$ if $a\in A$;
  \item $L^{\times}: S^\times \to 2^{\ap}$, where $L^{\times}((s, q))=L(s)$;
  \item $\acc^\times=\{(s, q)\in S^\times  \mid q\in \acc\}$.
\end{itemize}
\end{definition}

\section{Partially Observable Monte-Carlo Planning}\label{pomcp}
POMCP~\cite{silver2010monte} is a widely adopted online planning algorithm for POMDPs, as it effectively mitigates the \emph{curse of dimensionality} through State sampling and the \emph{curse of history} via history sampling with a black-box simulator. At each step $t$, it uses Monte Carlo Tree Search~\cite{coulom2006efficient} to build and explore a search tree rooted at $\T(h_t) = \langle \N(h_t), \V(h_t), \P(h_t) \rangle$, where $\N(h_t)$ tracks the number of visits to the history $h_t$, $\V(h_t)$ estimates the expected return, and $\P(h_t)$ holds particles that approximate the belief State $b_t$. 

The algorithm mainly consists of four stages:
(1) \textbf{Selection:} A State $s$ is selected uniformly from the particle set $\P(h_t)$. 
(2) \textbf{Simulation:} If the current node is internal, an action $a$ is selected to maximise $\V(ha) + c \sqrt{\frac{\log N(h)}{N(ha)}}$ following the upper confidence bound rule, balancing exploration and exploitation. 
(3) \textbf{Expansion:} When a leaf node is reached, child nodes $\T(ha) = \langle \N(ha), \V(ha), \emptyset \rangle$ are created for all possible actions. An action is then chosen according to a rollout policy (e.g., uniform random), and a new State $s'$ is generated using a simulator. The resulting State is added to $\P(hao)$. This process continues until a predefined depth limit. 
(4) \textbf{Backpropagation:} After simulation, the statistics of the traversed nodes are updated based on the results.
The agent then picks the action $a_t = \arg \max_a \V(h_t a)$, obtains a new observation $o_{t+1}$, and initiates the next search tree from the updated root $\T(h_t a_t o_{t+1})$.

\section{Proof of Theorem~2}
\label{App: proof}
\begin{theorem}
Given the sub-BSMDP $\hat{\M}^{\times}_{B_{\supp}}$ and its BS-MECs $(\E^{\times}_{B_{\supp}}, A_{\E^{\times}})$, Algorithm~1 returns a sound and complete union set $\hat{\W}^\times$ over all BS-AMECs.
\end{theorem}

\begin{proof} 
Let $(E, A_E) \in (\mathcal{E}^{\times}_{B_{\text{supp}}}, A_{\mathcal{E}^{\times}})$ be a BS-MEC. We show that Algorithm~1 adds $(E,A_E)$ to $\hat{\W}^{\times}$ if and only if
$(E,A_E)$ is a BS-AMEC.

\textbf{Soundness:} 
Suppose $(E, A_E)$ is added to $\hat{\W}^\times$, then $(E, A_E)$ is a BS-AMEC.
\begin{itemize}
    \item \textbf{Case 1:} $\exists\, \Theta \in E$ such that $\Theta \subseteq \acc^\times$. 
    
    Consider a Round-Robin belief support policy $\tilde{\pi}^\times$ over $E$. Since the digraph induced by $(E, A_E)$ is strongly connected, the Markov chain induced by $\tilde{\pi}^\times$ is a single bottom strongly connected component (BSCC). By Theorem 10.27 in~\cite{baier2008principles},  every belief support in $E$ is visited infinitely often almost surely. Since $\Theta \subseteq \acc^\times$, all states $s \in \Theta$ are accepting states and are visited infinitely often, so $\Theta$ is a winning support. By strong connectivity, every other belief support in $E$ reaches $\Theta$ infinitely often almost surely. Thus all belief supports in $E$ are winning supports, and $(E,A_E)$ is a BS-AMEC.

    \item \textbf{Case 2:} $S^{\times\times}_w = S^{\times\times}$.

    In this case, we first show that there exits a randomised belief support policy $\bar{\pi}^\times : B^\times_{\supp} \to \dist(A^\times)$ that if $S^{\times\times}_w = S^{\times\times}$, then $\mathrm{Pr}_{(E, A_E)}^{\bar{\pi}^\times}(\ltlG\ltlF\,\acc^\times) = 1$. 

    We consider a randomised memoryless belief-support policy $\bar{\pi}^{\times}$ as follows: for each belief support $\Theta \in E$,
\[
\bar{\pi}^{\times}(\Theta)(a) :=
\begin{cases}
\dfrac{1}{|A_E(\Theta)|}, & \text{if } a \in A_E(\Theta),\\[4pt]
0, & \text{otherwise}.
\end{cases}
\]
That is, $\bar{\pi}^{\times}$ selects uniformly among the actions enabled at $\Theta$. Since the product MDP $\M^{\times\times}$ synchronises the transitions of the belief supports $\Theta \in E$ with the transitions of their underlying states $s^\times \in \Theta$ and $A_E((s^\times, \Theta)) = A_E(\Theta)$, for $\M^{\times\times}$, this policy induce the corresponding state-based policy $\bar{\pi}^{\times\times}((s^\times,\Theta))$ with $s^\times \in \Theta$ that uniformly selects the actions enabled at $(s^\times,\Theta)$  as
\[
\bar{\pi}^{\times\times}((s^\times,\Theta)) :=
\bar{\pi}^{\times}(\Theta)
\]

Because $S^{\times\times}_w=S^{\times\times}$, every state of
$\M^{\times\times}$ is winning under the action restriction $A_E$. Hence
every bottom MEC of $\M^{\times\times}$ is accepting, that is,
$\mathcal{E}^{\times\times}\cap\acc^{\times}\neq\emptyset$.
The Markov chain induced by $\bar{\pi}^{\times\times}$ has the same underlying graph as
$\M^{\times\times}$ restricted to $A_E$.
Therefore, the BSCCs $T$ in $\M^{\times\times}_{\bar{\pi}}$ has $T \cap \acc \not= \emptyset$, i.e., $\mathrm{Pr}_{(E, A_E)}^{\bar{\pi}^{\times\times}}(\ltlG\ltlF\,\acc^\times) = 1$. Consequently, $\mathrm{Pr}_{(E, A_E)}^{\bar{\pi}^\times}(\ltlG\ltlF\,\acc^\times) = 1$. 
By Lemma 1 in~\cite{chatterjee2010randomness}, for this randomised
belief-support policy there exists a deterministic belief-support policy
$\tilde{\pi}^{\times}:B^{\times}_{\supp}\to A^{\times}$ such that
$
\Pr_{(E,A_E)}^{\tilde{\pi}^{\times}}
(\ltlG\ltlF\,\acc^{\times})=1 .
$
Thus $(E,A_E)$ is a BS-AMEC.

\end{itemize}

\textbf{Completeness.}
Suppose $(E, A_E)$ is a BS-AMEC, then $(E, A_E)$ is added to $\hat{\W}^\times$.
By Definition~7, every belief support $\Theta\in E$ is winning under some belief-support policy $\tilde{\pi}^{\times}$.
If there exists $\Theta\in E$ such that $\Theta\subseteq\acc^{\times}$, then Case~1 holds. 
Otherwise, since every $\Theta\in E$ is winning, Definition~5 implies that, under the corresponding policy, every state $s^{\times}\in\Theta$ satisfies the LTL objective with probability one. Hence the winning-set computation on $\M^{\times\times}$ results $S^{\times\times}_w=S^{\times\times}$. Therefore Algorithm~1 adds $(E,A_E)$ by Case~2.
Hence, Algorithm~1 is sound and complete.
\end{proof}

\section{Implementation Details}
\subsection{Algorithm 1}
For the winning state set $S^{\times\times}_{w}$ of the product MDP $\M^{\times\times}$, we first identify all AMECs using the standard procedure, for example Algorithm~47 in~\cite{baier2008principles}. We then compute the backward reachability set of these AMECs following Algorithm~45 in~\cite{baier2008principles}. If every state can reach some AMEC, then all states are almost surely winning in $\M^{\times \times}$.

\subsection{POMCP for LTL}
\label{ap-pomcp}
For reward assignment, each simulation (or rollout) propagates a particle $s^\times_{t+k}$ along with its belief support $\Theta(h_{t+k})$. An action $a^\times_{t+k+1}$ is sampled from $A_E(\Theta(h_{t+k}))$, the state transitions to $s^\times_{t+k+1}$, and an observation $o^\times_{t+k+1}$ is generated. The belief support $\Theta(h_{t+k+1})$ is then updated according to Eq.~(1) of the BSMDP model, up to a fixed depth or belief-support-based AMECs are reached. A subtlety in the product POMDP arises from the use of LDBAs, which may include the $\epsilon$-action $\epsilon_{q_i}$. To handle $\epsilon_{q_i}$, the belief support is updated directly as $\Theta(h_{t+1}) = \left\{ (s_t, q_i) \mid (s_t, q_t) \in \Theta(h_t) \right\}$.

For termination, we include each winning support in the certified partially winning set as a terminal action. If the UCB rule selects such an action and the sampled state lies within the corresponding winning support, the simulation and the rollout both terminate immediately with a return of 1.

\begin{algorithm}

\caption{POMCP for LTL}
\begin{algorithmic}[1]
\State \textbf{procedure} \textsc{Search}($h_t$)
\Repeat
    \If{$t = 0$}
        \State $s \sim b_0$
    \Else
        \State $s \sim \P(h_t)$
    \EndIf
    \State \textsc{Simulate}($s, h_t, 0$)
\Until{\textsc{Timeout}()}
\Return $\arg\max_a V(h_ta)$
\State \textbf{end procedure}
\State
\State \textbf{procedure} \textsc{Simulate}($s, h_{t}, k, \Theta(h_{t})$)
\If{$\Theta(h_t) \in \hat{\W}^\times$}
    \Return 1
\EndIf
\If{$k \geq k_{\max}$}
    \Return 0
\EndIf
\If{$h_t \notin \T$}
    \ForAll{$a \in A_E(\Theta_t)$}
        \State $\T(h_ta) \leftarrow (0, 0, \emptyset, \emptyset)$ 
        \\  $\T_{\text{init}}(h_ta) = (\N_{\text{init}}(h_ta), \V_{\text{init}}(h_t a), \P_{\text{init}}(h_t a), \Theta_{\text{init}}(h_ta))$
    \EndFor
    \Return \textsc{Rollout}($s, h_t, k, \Theta(h_t)$)
\EndIf
\State $a \leftarrow \arg\max_a V(ha) + c \sqrt{\frac{\log N(h)}{N(ha)}}$
\If{ $a \in \hat{\W}^\times$ and $s \in a$}
    \Return 1
\EndIf
\If{ $a \in \hat{\W}^\times$ and $s \not \in a$}
    \Return 0
\EndIf
\State $(s', o) \sim G(s, a)$
\State $h_{t+1} \leftarrow h_tao$
\State $\Theta(h_{t+1}) \leftarrow  T_{B_{\supp}}(\Theta(h_t), a)( o)$
\State $R \leftarrow \textsc{Simulate}(s', h_{t+1}, k+1, \Theta(h_{t+1}))$
\State $\P(h) \leftarrow \P(h) \cup \{s\}$
\State $\N(h) \leftarrow \N(h) + 1$
\State $\N(ha) \leftarrow \N(ha) + 1$
\State $\V(ha) \leftarrow \V(ha) + \frac{R - \V(ha)}{\N(ha)}$ \\
\Return $R$
\State \textbf{end procedure}
\State
\State \textbf{procedure} \textsc{Rollout}($s, h_t, k, \Theta(h_t)$)
\If{$\Theta(h_t) \in \hat{\W}^\times$}
    \Return 1
\EndIf
\If{$k \geq k_{\max}$}
    \Return 0
\EndIf
\State $a \sim \pi_{\text{rollout}}(h_t, \cdot)$
\If{ $a \in \hat{\W}^\times$ and $s \in a$}
    \Return 1
\EndIf
\If{ $a \in \hat{\W}^\times$ and $s \not \in a$}
    \Return 0
\EndIf
\State $(s', o) \sim G(s, a)$
\State $h_{t+1} \leftarrow h_tao$
\State $\Theta(h_{t+1}) \leftarrow  T_{B_{\supp}}(\Theta(h_t), a)( o)$ \\
\Return \textsc{Rollout}($s', h_{t+1}, k+1, \Theta(h_{t+1})$)
\State \textbf{end procedure}
\end{algorithmic}
\label{alg:pomcp}
\end{algorithm}

The overall algorithm is given in Algorithm~\ref{alg:pomcp}.

\section{Experiment Details}
\label{ap-exp}
\subsection{Experimental Setup} 
We use the following hyperparameters for POMCP: 30,000 simulations per planning step, a simulation depth of 200 (increased to 250 for the Grid-120 setting), 10,000 particles sampled from the initial state distribution, and a UCB exploration constant of 1. All experiments were run on a workstation equipped with an Intel Xeon Gold 6230R CPU (2.10 GHz), restricted to a single thread with up to 50 GB RAM, running Ubuntu 20.04.6 LTS. All source code is provided in the supplementary material.

\subsection{Toy Example}

\label{ap:exp}
\begin{verbatim}
Winning supports: 
    {"['s8_0']"}
Winning supports sets: 
    {frozenset({'s8_0'})}
Partically Winning Supports: 
    {"['s7_0', 's8_0']": 
        [['s8_0']], 
    "['s7_sink', 's8_0']": 
        [['s8_0']], 
    "['s7_0', 's7_sink', 's8_0']": 
        [['s8_0']], 
    "['s8_0', 's8_sink']": 
        [['s8_0']], 
    "['s7_0', 's8_0', 's8_sink']": 
        [['s8_0']], 
    "['s7_sink', 's8_0', 's8_sink']": 
        [['s8_0']], 
    "['s7_0', 's7_sink', 's8_0', 
    's8_sink']": 
        [['s8_0']]}
AMECs and Actions: 
    ({frozenset({"['s8_0']"})}, 
    {"['s8_0']": {'r', 'd'}})
Time step 1 Selected action: r 
    Value: 0.7655767387368196 
    Current belief:  {'s0_0': 1000}
    Sampled State: s0_0
Time step 2 Selected action: d 
    Value: 0.7853010164190775 
    Current belief: {'s1_0': 1000} 
    Sampled State: s1_0
Time step 3 Selected action: d 
    Value: 0.7900195761926022 
    Current belief: {'s3_0': 1000} 
    Sampled State: s3_0
Time step 4 Selected action: l 
    Value: 0.7941048537727401 
    Current belief:  {'s6_0': 1000} 
    Sampled State: s6_0
Time step 5 Selected action: d 
    Value: 0.7945807309601273 
    Current belief: {'s5_0': 1000} 
    Sampled State: s5_0
Time step 6 Selected action: {'s8_0'} 
    Value: 0.7972816433701444 
    Current belief: {'s8_0': 792, 
                    's7_sink': 208} 
    Sampled State: s8_0
chose belief support policy, 
stop simulation, 
winning support: {'s8_0'}
\end{verbatim}

\begin{table*}
\centering
\begin{tabular}{cccccccccc}
\toprule
Env. & Set. & $|S| / |O| / |T|$& Prop. & $|\A_\varphi|$ & $|S^\times| / |O^\times| / |T^\times|$ & SR (s) & $\Pr_{\M}^\pi(\varphi)$ & POMCP (s) \\

\midrule
\multirow{4}{*}{Hallway} 
 & HW1 & 60 / 24 / 208 & $\varphi_1$ & 3 & 180 / 24 / 624 & 0.015 & 0.254 & 74.967\\
 &  && $\varphi_2$ & 3 & 180 / 24 / 684 & 0.020 & 0.707 & 52.759\\
 & HW2 &92 / 24 / 325& $\varphi_1$ & 3 & 276 / 24 / 975  & 0.547 & 0.247 & 64.356 \\
 &  && $\varphi_2$ & 3 & 276 / 24 / 1067 & 0.743 & 0.464 & 44.643\\

\midrule
 \multirow{6}{*}{Rock Sample} 
 & 4 &256 / 112 / 1984& $\varphi_3$ & 4 & 1024 / 112 / 7936 & 42.231 & 0.453 & 79.102 \\
 &  && $\varphi_4$ & 3 & 768 / 112 / 6208  & 0.078 & 0.469 & 91.788 \\
 
 & 5 &400 / 132 /3120& $\varphi_3$ & 4 & 12480 / 132 / 11200 & 58.246 & 0.326 & 121.315\\
 &  && $\varphi_4$ & 3 & 1200 / 132 / 9760  & 0.114  & 0.455 & 120.307\\
 
 & 7 &784 / 244 / 6160& $\varphi_3$ & 4 & 3136/ 244 / 24640 & 105.475 & 0.180 & 129.259 \\
 &  && $\varphi_4$ & 3 & 2352 / 244 / 19264 & 0.210 & 0.003 & 160.291 \\

\midrule
\multirow{6}{*}{Grid} 
 & 10 &100 / 50 / 570& $\varphi_5$ &3& 300 / 50 / 1710 & 0.892 & 0.689 & 29.396 \\
 &  && $\varphi_6$ &3& 300 / 50 / 1710 & 0.895 & 1.0 & 0.244 \\
 &  && $\varphi_7$ &2& 200 / 50 / 1140 & 0.024 & 0.726 & 24.464 \\

 & 50 &2500 / 1250 / 10330& $\varphi_5$ &3& 7500 / 1250 / 30990 & 10.987 & 0.334 & 82.198\\
 &  && $\varphi_6$ &3& 7500 / 1250 / 30990 & 10.597 & 1.0 & 0.246 \\
 &  && $\varphi_7$ &2& 5000 / 1250 / 20660 & 1.843 & 0.530 & 67.959 \\

 & 100 &10000 / 5000 / 40280& $\varphi_5$ &3& 30000 / 5000 / 120840 & 81.361 &0.272 & 99.407 \\
 &  && $\varphi_6$ &3& 30000 / 5000 / 120840 & 82.334 & 1.0 & 0.248 \\
 &  && $\varphi_7$ &2& 20000 / 5000 / 80560 & 30.603 & 0.308 & 107.366 \\

 & 120 &14400 / 7200 / 57860& $\varphi_5$ &3& 43200 / 7200 / 173580 & 191.431 & 0.0 & 182.667\\
 &  && $\varphi_6$ &3& 43200 / 7200 / 173580 & 188.264 & 1.0 & 0.267 \\
 &  && $\varphi_7$ &2& 28800 / 7200 / 115720  & 66.314 & 0.058  & 183.521\\

\bottomrule
\end{tabular}
\caption{Numerical Results for POMDP Policy Synthesis under LTL Specifications.}
\label{res-table2}
\end{table*}
\subsection{Benchmark Domain}

\paragraph{Hallway}
We consider two Hallway models of different sizes, $HW_1$ and $HW_2$, originally introduced in~\cite{littman1995learning} and also used in~\cite{chatterjee2015qualitative}. In this setting, an agent navigates an office building with uncertain location and orientation (North, East, South, or West). The agent can move forward or turn left, each succeeding with some probability. Its short-range sensor provides information about whether it is adjacent to a wall in its current orientation, subject to a certain probability of error. Specific locations in the office are labelled $A$, $B$ and $C$. We consider the following LTL objectives:
\begin{itemize}
    \item $\varphi_1 = \ltlF A \land \ltlG \neg B$: a reach-and-avoid task where the agent must eventually reach $A$ while always avoiding $B$.
    \item $\varphi_2 = \ltlF \ltlG C$: a persistence task where the agent must reach $A$ and remain there indefinitely.
\end{itemize}

\paragraph{Rock Sample}
We follow the rock sample model from~\cite{junges2021enforcing}, where a robot operates on an $n \times n$ grid containing two rocks, each either good $G$ or bad $B$, with initially unknown quality. The robot can move between cells, use a sensor on adjacent cells to observe rock quality, or sample a rock to collect it. We consider the following LTL tasks:
\begin{itemize}
    \item $\varphi_3 = \ltlF(G \land \ltlF E) \land \ltlG \neg B$: a sequential task requiring the robot to first sample a good rock, then eventually reach exit States ($E$), while never sampling a bad rock.
    \item $\varphi_4 = \ltlF \ltlG G$: a persistence task in which the robot finally collects the good rock and remain at that location indefinitely.
\end{itemize}

\paragraph{Grid World}
The model $Grid(n, n)$ is based on the problem introduced in~\cite{littman1995learning}. It comprises an $n \times n$ grid where each cell is labelled as a trap ($T$), goal 1 ($G_1$), or goal 2 ($G_2$). The agent moves left, right, or down, with a 0.2 probability of movement failure included in the transition dynamics. An observation is defined for each of the neighbouring States. We consider the following LTL tasks:
\begin{itemize}
    \item $\varphi_5 = \ltlF G_1 \land \ltlG \neg T$: a reach-and-avoid task where the agent must eventually reach $G_1$ while always avoiding traps.
    \item $\varphi_6 = \ltlG (G_2 \rightarrow \ltlF G_1) \land \ltlG \neg T$: a reactive task where the agent must always avoid traps and, whenever it reaches $G_1$, must eventually reach $G_2$.
    \item $\varphi_7 = \ltlG \ltlF G_1 \land \ltlG \neg T$: a recurrence task where the agent must visit $G_1$ infinitely often while avoiding traps.
    
\end{itemize}

\subsection{Full Results}
\label{ap-res}
Table~\ref{res-table2} summarises the performance of our approach across all tasks. For each case, we report the benchmark domain (Env.) and its setting (set.), the original model details ($|S| / |O| / |T|$), the LTL objective ($\varphi_i$), the product model details ($|S^\times| / |O^\times| / |T^\times|$), the time to construct the sound reward function (SR (s)), the satisfaction probability ($\Pr_{\M}^\pi(\varphi)$), and the policy synthesis time using the enhanced POMCP solver (POMCP (s)).

%% file: main.bib
@book{bongard2008probabilistic,
  title={{Probabilistic Robotics}},
  author={Sebastian, Thrun and  Wolfram, Burgard and Dieter, Fox},
  year={2005},
  publisher={The MIT Press, Cambridge}
}

@inproceedings{pnueli1977temporal,
  title={The temporal logic of programs},
  author={Pnueli, Amir},
  booktitle={18th annual symposium on foundations of computer science (sfcs 1977)},
  pages={46--57},
  year={1977},
  organization={ieee}
}

@inproceedings{bouton2020point,
  title={Point-based methods for model checking in partially observable Markov decision processes},
  author={Bouton, Maxime and Tumova, Jana and Kochenderfer, Mykel J},
  booktitle={Proceedings of the AAAI Conference on Artificial Intelligence},
  volume={34},
  pages={10061--10068},
  year={2020}
}

@inproceedings{sharan2014finite,
  title={Finite state control of POMDPs with LTL specifications},
  author={Sharan, Rangoli and Burdick, Joel},
  booktitle={2014 American Control Conference},
  pages={501--508},
  year={2014},
  organization={IEEE}
}

@article{kalagarla2024optimal,
  title={Optimal Control of Logically Constrained Partially Observable and Multiagent Markov Decision Processes},
  author={Kalagarla, Krishna C and Kartik, Dhruva and Shen, Dongming and Jain, Rahul and Nayyar, Ashutosh and Nuzzo, Pierluigi},
  journal={IEEE Transactions on Automatic Control},
  volume={70},
  number={1},
  pages={263--277},
  year={2024},
  publisher={IEEE}
}

@inproceedings{pineau2003point,
  title={Point-based value iteration: An anytime algorithm for POMDPs},
  author={Pineau, Joelle and Gordon, Geoff and Thrun, Sebastian and others},
  booktitle={Ijcai},
  volume={3},
  pages={1025--1032},
  year={2003}
}

@article{silver2010monte,
  title={Monte-Carlo planning in large POMDPs},
  author={Silver, David and Veness, Joel},
  journal={Advances in neural information processing systems},
  volume={23},
  year={2010}
}

@incollection{littman1995learning,
  title={Learning policies for partially observable environments: Scaling up},
  author={Littman, Michael L and Cassandra, Anthony R and Kaelbling, Leslie Pack},
  booktitle={Machine Learning Proceedings 1995},
  pages={362--370},
  year={1995},
  publisher={Elsevier}
}

@inproceedings{chatterjee2015qualitative,
  title={Qualitative analysis of POMDPs with temporal logic specifications for robotics applications},
  author={Chatterjee, Krishnendu and Chmelik, Martin and Gupta, Raghav and Kanodia, Ayush},
  booktitle={2015 IEEE International Conference on Robotics and Automation (ICRA)},
  pages={325--330},
  year={2015},
  organization={IEEE}
}

@inproceedings{sickert2016limit,
  title={{Limit-deterministic B{\"u}chi automata for linear temporal logic}},
  author={Sickert, Salomon and Esparza, Javier and Jaax, Stefan and K{\v{r}}et{\'\i}nsk{\`y}, Jan},
  booktitle={International Conference on Computer Aided Verification},
  pages={312--332},
  year={2016},
  organization={Springer}
}

@book{baier2008principles,
  title={{Principles of model checking}},
  author={Baier, Christel and Katoen, Joost-Pieter},
  year={2008},
  publisher={MIT press}
}

@article{vardi1994reasoning,
  title={Reasoning about infinite computations},
  author={Vardi, Moshe Y and Wolper, Pierre},
  journal={Information and computation},
  volume={115},
  number={1},
  pages={1--37},
  year={1994},
  publisher={Elsevier}
}

@inproceedings{kvretinsky2018rabinizer,
  title={Rabinizer 4: From LTL to your favourite deterministic automaton},
  author={K{\v{r}}et{\'\i}nsk{\`y}, Jan and Meggendorfer, Tobias and Sickert, Salomon and Ziegler, Christopher},
  booktitle={in Proceedings of International Conference on Computer Aided Verification},
  pages={567--577},
  year={2018},
  organization={Springer}
}

@inproceedings{kaelbling1995partially,
  title={Partially observable markov decision processes for artificial intelligence},
  author={Kaelbling, Leslie Pack and Littman, Michael L and Cassandra, Anthony R},
  booktitle={International Workshop on Reasoning with Uncertainty in Robotics},
  pages={146--163},
  year={1995},
  organization={Springer}
}

@inproceedings{coulom2006efficient,
  title={Efficient selectivity and backup operators in Monte-Carlo tree search},
  author={Coulom, R{\'e}mi},
  booktitle={International conference on computers and games},
  pages={72--83},
  year={2006},
  organization={Springer}
}

@inproceedings{junges2021enforcing,
  title={Enforcing almost-sure reachability in POMDPs},
  author={Junges, Sebastian and Jansen, Nils and Seshia, Sanjit A},
  booktitle={International Conference on Computer Aided Verification},
  pages={602--625},
  year={2021},
  organization={Springer}
}

@article{chatterjee2016decidable,
  title={What is decidable about partially observable Markov decision processes with $\omega$-regular objectives},
  author={Chatterjee, Krishnendu and Chmelik, Martin and Tracol, Mathieu},
  journal={Journal of Computer and System Sciences},
  volume={82},
  number={5},
  pages={878--911},
  year={2016},
  publisher={Elsevier}
}

@article{chatterjee2007algorithms,
  title={Algorithms for omega-regular games with imperfect information},
  author={Chatterjee, Krishnendu and Doyen, Laurent and Henzinger, Thomas A and Raskin, Jean-Fran{\c{c}}ois},
  journal={Logical Methods in Computer Science},
  volume={3},
  year={2007},
  publisher={Episciences. org}
}

@inproceedings{baier2008decision,
  title={On decision problems for probabilistic B{\"u}chi automata},
  author={Baier, Christel and Bertrand, Nathalie and Gr{\"o}{\ss}er, Marcus},
  booktitle={International Conference on Foundations of Software Science and Computational Structures},
  pages={287--301},
  year={2008},
  organization={Springer}
}

@article{ahmadi2020stochastic,
  title={Stochastic finite state control of POMDPs with LTL specifications},
  author={Ahmadi, Mohamadreza and Sharan, Rangoli and Burdick, Joel W},
  journal={arXiv preprint arXiv:2001.07679},
  year={2020}
}

@article{madani1999undecidability,
  title={On the undecidability of probabilistic planning and infinite-horizon partially observable Markov decision problems},
  author={Madani, Omid and Hanks, Steve and Condon, Anne},
  journal={Aaai/iaai},
  volume={10},
  number={315149.315395},
  year={1999}
}

@inproceedings{liu2021leveraging,
  title={Leveraging temporal structure in safety-critical task specifications for POMDP planning},
  author={Liu, Jason and Rosen, Eric and Zheng, Suchen and Tellex, Stefanie and Konidaris, George},
  booktitle={Proc. RSS Workshop Robot. People},
  year={2021}
}

@article{carr2021task,
  title={Task-aware verifiable RNN-based policies for partially observable Markov decision processes},
  author={Carr, Steven and Jansen, Nils and Topcu, Ufuk},
  journal={Journal of Artificial Intelligence Research},
  volume={72},
  pages={819--847},
  year={2021}
}

@article{carr2020verifiable,
  title={Verifiable RNN-based policies for POMDPs under temporal logic constraints},
  author={Carr, Steven and Jansen, Nils and Topcu, Ufuk},
  journal={arXiv preprint arXiv:2002.05615},
  year={2020}
}

@inproceedings{kalagarla2022optimal,
  title={Optimal control of partially observable markov decision processes with finite linear temporal logic constraints},
  author={Kalagarla, Krishna C and Dhruva, Kartik and Shen, Dongming and Jain, Rahul and Nayyar, Ashutosh and Nuzzo, Pierluigi},
  booktitle={Uncertainty in Artificial Intelligence},
  pages={949--958},
  year={2022},
  organization={PMLR}
}

@inproceedings{chatterjee2010randomness,
  title={Randomness for free},
  author={Chatterjee, Krishnendu and Doyen, Laurent and Gimbert, Hugo and Henzinger, Thomas A},
  booktitle={International Symposium on Mathematical Foundations of Computer Science},
  pages={246--257},
  year={2010},
  organization={Springer}
}
